\documentclass[10pt,twocolumn,oneside]{IEEEtran}
\usepackage{,amssymb,paralist,ifthen}
\usepackage{amsmath,amsfonts,amssymb,bm}
\usepackage[tight,footnotesize]{subfigure}
\usepackage{graphicx,stfloats}
\usepackage{algorithm,algpseudocode}
\usepackage{cite,url}
\usepackage{color}
\usepackage{theorem} 

\newtheorem{lemma}{Lemma}

\newtheorem{Theorem}{Theorem}
\newtheorem{Def}{Definition}

\begin{document}

\title{Distributed Power Allocation for \\Coordinated
Multipoint Transmissions in \\Distributed Antenna Systems}

\author{Xiujun~Zhang,~\IEEEmembership{Member,~IEEE,}
        Yin~Sun,~\IEEEmembership{Member,~IEEE,}~Xiang~Chen,~\IEEEmembership{Member,~IEEE,}
        Shidong~Zhou,~\IEEEmembership{Member,~IEEE,}
        Jing~Wang,~\IEEEmembership{Member,~IEEE,
        and Ness B. Shroff ,~\IEEEmembership{Fellow,~IEEE}
\thanks{Manuscript received June 14, 2012; revised October 16, 2012 and February 3, 2013; accepted February 17, 2013. The associate editor coordinating the review of this paper and approving it for publication was Shahrokh Valaee.}
\thanks{This work is partially supported by National Basic Research Program of China (2012CB316002), China's S\&T Major Project (2011ZX03003-001-01), National Natural Science Foundation of China (61132002), Huawei Corp., International Science and Technology Cooperation Program (2010DFB10410), Tsinghua Research Funding (2010THZ02-3), the Army Research Office MURI Program W911NF-08-1-0238, and NSF Grants CNS-1012700 and CNS-1065136.}
\thanks{$^\dag$Xiujun Zhang, Xiang Chen~(the corresponding author), Shidong Zhou and Jing Wang are with the Department of Electronic Engineering, Research Institute of Information Technology, Tsinghua National Laboratory for Information Science and Technology (TNList), Tsinghua University, Beijing, China, 100084, e-mail:~\{zhangxiujun,chenxiang,zhousd,wangj\}@tsinghua.edu.cn. Dr. Xiang Chen is also with Aerospace Center, Tsinghua University.}
\thanks{$^\ddag$Yin Sun is with the Department of Electrical and Computer Engineering, the Ohio State University, 2015 Neil Ave., Columbus, OH 43210, USA, e-mail: sunyin02@gmail.com.}
\thanks{$^\star$Ness B. Shroff is with the Department of Electrical and Computer Engineering and the Department of Computer Science and Engineering, the Ohio State University, 2015 Neil Ave., Columbus, OH 43210, USA, e-mail: shroff@ece.osu.edu.}}
\thanks{Digital Object Identifier 10.1109/TWC.2013.13.120863}}

\markboth{ACCEPTED BY IEEE TRANSACTIONS ON WIRELESS COMMUNICATIONS, FEB. 2013}%
{ACCEPTED BY IEEE TRANSACTIONS ON WIRELESS COMMUNICATIONS, FEB. 2013}

\maketitle

\begin{abstract}

This paper investigates the distributed power allocation problem for
coordinated multipoint (CoMP) transmissions in distributed antenna
systems (DAS). Traditional duality-based optimization techniques cannot be directly applied to this problem, because the non-strict concavity of the CoMP transmission's achievable rate with respect to the transmission power induces that the local power allocation subproblems have non-unique optimum solutions. 
We propose a distributed power allocation algorithm to resolve this
non-strict concavity difficulty. This algorithm only requires local
information exchange among neighboring base stations serving the
same user, and is thus flexible with respect to network size and
topology. The step-size parameters of this algorithm are determined
by only local user access relationship (i.e., the number of users
served by each antenna), but do not rely on channel coefficients.
Therefore, the convergence speed of this algorithm is quite robust
to channel fading. We rigorously prove that this algorithm converges
to an optimum solution of the power allocation problem. Simulation
results are presented to demonstrate the effectiveness of the
proposed power allocation algorithm.
\end{abstract}


\section{Introduction}\label{sec:intro}
The explosive growth of mobile access services has led to a huge demand for enhanced throughput and extended coverage in the next generation wireless networks. In recent years, distributed antenna system (DAS) has emerged as a promising network architecture to achieve these goals \cite{Li2011,Choi2007,Park2009}. In this architecture, each base station is equipped with some remote antennas which are distributed in the entire cell area, as shown in Fig.~\ref{fig1}.
These distributed antennas are connected to the base station via wired backhaul network. By this, nearby distributed antennas are able to coordinate with each other and provide enhanced service experience to the mobile users. This technique is called the
coordination multipoint (CoMP) transmission in the literature \cite{Park2009,Sawahashi2010}.

One of the key techniques to realize high throughput in wireless networks is power allocation. Traditionally, power allocation of wireless networks is handled by centralized algorithms, e.g., \cite{Venturino2009,Papandreou2008,Gesbert2010,Luo2011}. These algorithms
request {multi-hop} signaling mechanisms to gather the channel state information (CSI) of all the wireless links at a central processing
unit {in a short time period}, and then distribute the obtained power allocation solution to the transmitters. Such mechanisms would generate enormous signaling overhead on the backhaul network, and is probably not scalable when the network size grows large.

Recently, a great deal of research efforts have focused on
distributed power allocation for various wireless networks. Game
theory based power allocation techniques were proposed in
\cite{Yu2002,Yu2004,Huang2006,Gesbert2007,Pang2008}, which intend to
compute Nash equilibrium power allocation solutions. However, these
Nash equilibrium solutions might be far from optimality
\cite{Gesbert2007}. Duality-based distributed power allocation
techniques were proposed in
\cite{Xiao04,LinJsac06,Chiang07ProcIEEE}, where the global power
allocation problem is decomposed into many local power allocation
subproblems, each of which can be solved by utilizing only locally
available network information. However, these techniques cannot be
directly applied to CoMP transmissions in DAS --- the local power
allocation subproblem may have many optimum solutions, because the
data rate of CoMP transmission is not strictly concave {with respect to} the power
variables \cite{Xiao04,Bertsekas1999}. Since no global network information is available when
solving the local power allocation subproblems, it is quite difficult
to find a global feasible solution among all the locally optimum
solutions.

One promising method to address this non-strict concavity problem is
the proximal point method \cite{Bertsekas1989}, which adds strictly
concave terms to the objective function without affecting the
optimum solution. However, typical proximal point algorithms require
a two-layer nested iteration structure, where each outer-layer
update can proceed only after the inner-layer iterations converge
\cite{Bertsekas1989}. Such a structure is not suitable for on-line
distributed implementation, because it is difficult to decide in a
distributed manner when the inner-layer iterations can stop. In
\cite{Lin2006}, a single-layer proximal point algorithm was proposed
for multi-path routing problems. However, the convergence analysis
in \cite{Lin2006} {also} cannot be directly utilized for the power
allocation problem considered here, owing to the additional channel
coefficients in our problem. It is difficult to answer whether the
channel coefficients have significant impact on the convergence
behavior of the algorithm mentioned above.

\begin{figure}[t]
    \centering
    \scalebox{0.7}{\includegraphics*[0,0][341,242]{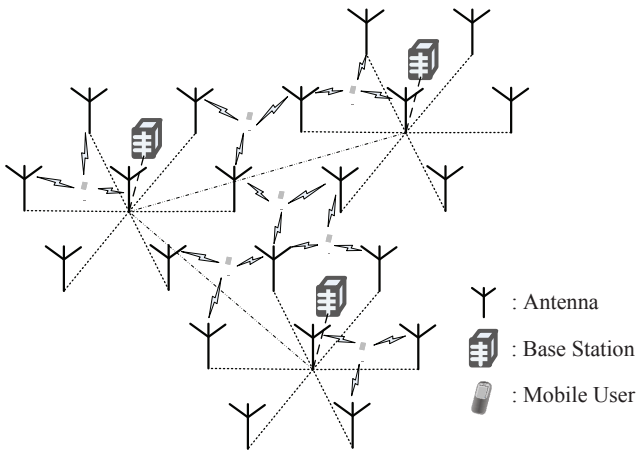}}
    \caption{System model of coordination multipoint (CoMP) transmissions in a distributed antenna system (DAS).} \label{fig1}
    \vspace{-10pt}
\end{figure}

This paper investigates the distributed power allocation problem for
a downlink DAS with many antennas and many single antenna users.
Each user is served by several nearby antennas via CoMP transmission
techniques. Meanwhile, each antenna may serve several users over
orthogonal channels. The main contributions of this paper are
summarized as follows:
\begin{itemize}\item [1)] A distributed power allocation algorithm is proposed to maximize {the weighted sum rate} of the downlink DAS, subject to per-antenna power consumption constraints. This power allocation algorithm is implemented distributedly
among the base stations instead of being executed in a centralized
fashion. The algorithm possesses a nice single-layer iteration
structure, which is desirable for on-line implementations. {In each
iteration, the algorithm only requires local information exchange
among neighboring base stations serving the same user, which is
flexible with respect to network size and topology.}

\item [2)] A novel procedure is proposed to compute the primal optimum solution of the local power allocation subproblem, which is simpler than that proposed in \cite{Lin2006}.

\item [3)] We rigorously analyze the convergence and optimality of the proposed distributed algorithm for the power allocation problem. The bounds of the step-size parameters to ensure convergence are derived.

\item [4)] We show that the step-size parameters of this algorithm are determined by only local user access relationship (i.e., the number of users served by each antenna), but do not rely on channel coefficients\footnote{The channel coefficients are only utilized locally to solve the local power allocation subproblem.}. Therefore, the convergence speed of this algorithm is quite robust to channel fading.
\end{itemize}

Our proposed power allocation algorithm is motivated by the work of
\cite{Lin2006}. However, our work differs from it in several respects.
First, our analysis indicates that a larger step-size can be used for the algorithm in \cite{Lin2006}, which can achieve a faster convergence speed. Second, while our problem has additional channel fading coefficients, we show that the step-size parameters and the convergence speed of our algorithm are robust to different channel fading coefficients. Finally, our procedure for solving the local power allocation problem is simpler than that proposed in \cite{Lin2006}.

For ease of later use, we define the following notations: Let $|S|$ denote the number of elements in set $S$, and let $S/T$
denote the set $S/T=\{x|x\in S, x\notin T\}$. The projection of a
real number $x$ on the set $[0,\infty)$ is defined as $[x]^+ =
\max\{x,0\}$.

The remaining parts of this paper are organized as follows: In
Section \ref{sec:system_model}, the system model and problem
formulation are presented. Section \ref{sec:proposed} presents the
proposed power allocation algorithm and its distributed
implementation. Simulation results of the proposed power allocation
strategy are presented in Section IV. Conclusions are drawn in
Section V.

\section{System Model and Problem Formulation} \label{sec:system_model}
\subsection{System Model}
{We consider} a downlink DAS with $K$ distributed antennas and $N$ single antenna mobile users, which are denoted by $\mathcal {K}=\{1,2,\cdots,K\}$ and $\mathcal {N}=\{1,2,\cdots,N\}$, respectively. Each base station is equipped with several distributed antennas, as illustrated in Fig.~\ref{fig1}. These distributed antennas are connected to the base station via wired backhaul network.
The total throughput of this network is limited by the strong co-channel interference. By allowing several nearby antennas to transmit to one user in a coordinated fashion, the CoMP transmission techniques, such as space-time block coding or maximum
ratio transmission \cite{Park2009}, convert the strong interferences into useful signals and thereby significantly boost the total network throughput. The set of antennas serving the $n$th user is denoted by
$R(n)\subseteq\mathcal {K}$, and the set of users communicating
with the $k$th antenna is expressed as $U(k)\subseteq\mathcal {N}$. In practice, the number of serving antennas for each user, i.e., $|R(n)|$, is usually small, due to the limitation of implementation complexity for CoMP transmissions.

When the density of the distributed antennas is high, CoMP transmissions can not mitigate all the strong interferences, {which results in some strong residual interferences}.
In \cite{Etkin2007}, it was shown that orthogonal transmission is Pareto optimal for strong interference Gaussian channels. Therefore, the users with strong mutual interference should be scheduled to communicate over orthogonal channels via frequency (or time) division multiple access, while geographically separated users with weak mutual interference are allowed to share the same channel resource. This scheduling {task} belongs to the type of timetabling problem, which is a classic problem in the computer science literature with many practical algorithms \cite{Schaerf1999,Qu2009}.

After selecting proper antennas for CoMP transmission and scheduling the users, there are only weak interferences in the network. We consider a slow fading wireless environment.
Let $h_{kn}$ be the complex coefficient of the wireless channel from
the $k$th antenna to the $n$th user and $p_{kn}$ be the transmission power of the $k$th antenna for serving the $n$th user.
The data rate of CoMP transmission to the $n$th user is given by
\begin{equation} \label{eq0}
\mathbf{C}_{n}=\log_2\left(1+\frac{\sum_{k \in R(n)}{}
|h_{kn}|^2p_{kn}}{\sigma^2+\sum_{(k,m)\in I(n)}
|h_{kn}|^2p_{km}}\right),
\end{equation}
where $I(n)$ is the set of source antenna and serving user pair
which may interfere the $n$th user, or more specifically,  $(k,m)\in
I(n)$ represents that the $k$-th source antenna  serving the $m$-th
user through the same serving channel of the $n$th user.

There are two difficulties for utilizing the data rate function
$\mathbf{C}_{n}$ to formulate the power allocation problem: First,
it leads to a non-convex optimization problem that is NP-hard
\cite{Luo2008}, for which one may not be able to find a solution
that is both fast and global optimal even by centralized
optimization. The design of a distributed optimization algorithm
will be even more difficult, if not impossible. In order to reduce
the solution complexity, we need to find an approximate rate
function of $\mathbf{C}_{n}$ that is convex. Second, it can be quite
difficult to attain the exact expression of the rate function
$\mathbf{C}_{n}$. In practice, the number of interfering antennas is
usually much larger than the number of source antennas. Although the
receiver can get an accurate estimation of the channel gain
$|h_{kn}|^2$ for each source antenna $k\in R(n)$, it may be too
demanding to estimate the channel gain from the enormous interfering
antennas, especially when the powers of the interference signals are
weak. On the other hand, estimating the noise-plus-interference
power $\sigma_n^2 \triangleq\sigma^2+\sum_{(k,m)\in I(n)}
|h_{kn}|^2p_{km}$ is obviously much easier.
For these reasons, we consider to utilize an upper bound of the noise-plus-interference power $\sigma_{n}^2$, which is denoted by $\sigma_{n,peak}^2$, to derive an approximate rate function. Let $\gamma_{kn}=|h_{kn}|^2/{\sigma_{n,peak}^2}$
denote the normalized channel gain from the $k$th antenna to the
$n$th user. Then, we derive a conservative rate function
\begin{equation} \label{eq1}
\tilde{\mathbf{C}}_{n}=\log_2\left(1+\sum_{k \in R(n)} p_{kn} \gamma_{kn}\right)\leq {\mathbf{C}}_{n}.
\end{equation}
The key benefit of the conservative rate function
$\tilde{\mathbf{C}}_{n}$ is that it is convex and is computable
without accurate knowledge of the channel gain $|h_{kn}|^2$ for the enormous interfering antennas, which resolves the two difficulties mentioned above. We will illustrate the rate loss for using this conservative rate function to formulate the power allocation problem in Section \ref{sec:simulations}.

\subsection{Problem Formulation}
The rest of this paper focuses on the following power allocation problem to maximize
the {weighted sum rate} of the DAS:
\begin{eqnarray}\label{eq3}
\max_{p_{kn}} && \sum_{n=1}^{N} {w_n}\log_2\left(1+\sum_{k \in R(n)} p_{kn} \gamma_{kn}\right)\\
{\rm s.t.}&&\sum_{n\in U(k)}  p_{kn}\leq P_{k},~k=1,2,\ldots,K,\nonumber\\
&& p_{kn}\geq 0, ~k=1,2,\ldots,K,~n\in U(k), \nonumber
\end{eqnarray}
where {${w_n}>0$ is the weight of the $n$th user's data rate} and $P_{k}$ is the maximal allowable transmission power of the $k$th antenna.

This problem is a convex optimization problem, which can be solved
by standard centralized convex optimization algorithms such as the
interior point method \cite{Boyd2004}. However, these
centralized algorithms are hard to be fulfilled in large-scale DAS, due to the heavy signaling overhead over the backhaul network.
In contrast, duality-based optimization techniques
\cite{Xiao04,LinJsac06,Chiang07ProcIEEE} cannot be directly applied
to this problem, either, because they require the objective function to be strictly concave.
However, the objective function in \eqref{eq3}
is not strictly concave with respect to the transmission power variables, since
it is constant when the value of $\sum_{k \in R(n)} p_{kn} \gamma_{kn}$ is fixed.
If the duality-based optimization techniques
\cite{Xiao04,LinJsac06,Chiang07ProcIEEE} are utilized, the decomposed local power allocation subproblem may have many locally
optimum solutions at some special dual points. It is quite difficult to recover a global feasible solution among all the locally optimum solutions.
When the dual variables are updated around these dual points, the primal power allocation variables keep oscillating and hardly converge (see \cite{Lin2006} for more details).

\section{Distributed Power Allocation Algorithm}\label{sec:proposed}
In this section, we propose a power allocation algorithm to solve the problem (\ref{eq3}), which is distributed among the base stations instead of being centralized over the entire network. The key feature of this algorithm is that its step-size parameters and convergence speed are robust to different channel fading coefficients, which makes our algorithm quite convenient for practical implementations. The details are provided in the following subsections.
\subsection{Single-layer Distributed Power Allocation Algorithm}
To circumvent the aforementioned oscillation
problem, we make use of the idea in the proximal point method
\cite{Bertsekas1989}, which is to add some quadratic terms to the
objective function and make it strictly concave in the primal
variables.
We reformulate the original power allocation problem (\ref{eq3}) as
\begin{align}\label{eq6}
\max_{p_{kn},y_{kn}}~ & \sum_{n=1}^{N}{w_n} \log_2(1+\sum_{k \in R(n)} p_{kn} \gamma_{kn})\nonumber\\
& -\sum_{n=1}^{N}\sum_{k \in R(n)}\frac{c_n}{2}(p_{kn}-y_{kn})^2\\
{\rm s.t.}~~~ &\sum_{n\in U(k)}  p_{kn}\leq P_{k},~k=1,2,...,K,\label{eq5}\\
& p_{kn}\geq 0, ~k=1,2,...,K,~n\in U(k), \nonumber
\end{align}
where we have introduced some quadratic auxiliary terms to make the
objective function strictly concave with respect to the transmission
power variables. Here, $y_{kn}$ is the auxiliary variable
corresponding to $p_{kn}$, $c_n>0$ is the parameter of the quadratic
terms. For notational convenience, let us use {the $|R(n)|$
dimensional vector} $\vec{p}_{n}$ to denote the transmission power
variables of the antennas serving the $n$th users, and {the
$\sum_{n=1}^N|R(n)|$ dimensional vector}
$\vec{p}=[\vec{p}_{1}^T,\vec{p}_{2}^T,\cdots,\vec{p}_{N}^T]^{T}$ to
denote all the transmission power variables. Similarly, we define
{the $|R(n)|$ dimensional vector} $\vec{y}_{n}$ and {the
$\sum_{n=1}^N|R(n)|$ dimensional vector}
$\vec{y}=[\vec{y}_{1}^T,\vec{y}_{2}^T,\cdots,\vec{y}_{N}^T]^{T}$ as
the auxiliary variable vectors corresponding to $\vec{p}_{n}$ and
$\vec{p}$. It is known that the optimum value of the objective
function in \eqref{eq6} coincides with that in (\ref{eq3})
\cite{Bertsekas1989}. In particular, if $\vec{p}^*$ is the optimum
solution to \eqref{eq3}, then $\vec{p}=\vec{p}^*,\vec{y}=\vec{p}^*$
solves (\ref{eq6}).

The standard proximal point method in general has a two-layer
nested optimization structure: the inner layer iterations optimizing
$\vec{p}$ for fixed $\vec{y}$ by a Lagrangian dual optimization
method, and the outer layer iterations optimizing the auxiliary variable $\vec{y}$.
Such a layered structure is not suitable for on-line distributed implementations,
because it is difficult to decide in a distributed manner when the
inner-layer iterations have converged. In the following, we present
a modified proximal point method with a single-layer optimization
structure, where the outer-layer update of $\vec{y}$ does not
request that the inner-layer dual updates have converged.

The Lagrangian of the problem (\ref{eq6}) can be written as:
\begin{align}\label{eq7}
L(\vec{p},\vec{\lambda},\vec{y})=&\sum_{n=1}^{N}{w_n}\log_2(1+\sum_{k
\in R(n)} p_{kn} \gamma_{kn})\nonumber\\
&-\sum_{k=1}^{K}\lambda_{k}(\sum_{n
\in U(k)}p_{kn}-P_{k})\nonumber\\
&-\sum_{n=1}^{N}\sum_{k \in
R(n)}\frac{c_n}{2}(p_{kn}-y_{kn})^2,
\end{align}
where
$\vec{\lambda}=[{\lambda}_{1},{\lambda}_{2},\cdots,{\lambda}_{K}]^{T}$
is the vector of dual variables corresponding to the constraints in
\eqref{eq5}.
Now we are able to present our distributed power allocation
algorithm as the following:

\noindent\textbf{Algorithm A: Single-layer Distributed Power Allocation Algorithm}\\
\noindent At the $t$th iteration,
\begin{itemize}
\item [Step~1:] Dual variable update: \\
    Let $\vec{y}=\vec{y}(t)$ and $\vec{\lambda}=\vec{\lambda}(t)$, maximize $L(\vec{p},\vec{\lambda},\vec{y})$ with respect to $\vec{p}$:
\begin{equation}\label{eq8}
\vec{p}(t)={\textrm{arg}\max}_{\vec{p}\geq 0}
L(\vec{p},\vec{\lambda}(t),\vec{y}(t)).
\end{equation}
    Update the dual variables by
\begin{equation}\label{eq9}
\lambda_{k}(t+1)=[\lambda_{k}(t)+\alpha_k(\sum_{n \in
U(k)}p_{kn}-P_{k})]^{+},
\end{equation}
where $\alpha_k$ is the step-size of the dual update.

\item [Step~2:] Auxiliary variable update: \\
     Let $\vec{y}=\vec{y}(t)$ and $\vec{\lambda}=\vec{\lambda}(t+1)$, maximize $L(\vec{p},\vec{\lambda},\vec{y})$ with respect to $\vec{p}$:
 \begin{equation}\label{eq10}
\vec{z}(t)={\textrm{arg}\max}_{\vec{p}\geq 0}
L(\vec{p},\vec{\lambda}(t+1),\vec{y}(t)).
\end{equation}
    Update the auxiliary variables by
\begin{equation}\label{eq11}
y_{kn}(t+1)=y_{kn}(t)+\beta(z_{kn}(t)-y_{kn}(t)),
\end{equation}
where $0<\beta \leq 1$ is the step-size for auxiliary variable update.
\end{itemize}

The value of $\beta$ can be chosen arbitrarily in $(0,1]$. The choices of $\alpha_k$ to ensure convergence of \textbf{Algorithm A} will be discussed in Section \ref{sec1}. We note that while the convergence analysis in \cite{Lin2006} apply for the degenerated case of $\gamma_{kn}=1$, it is difficult to answer if practical channel coefficients $\gamma_{kn}$ would have significant impact on the convergence behavior of the algorithm. One major contribution of this paper is to show that \emph{the step-size parameters $\alpha_k$ to ensure convergence are irrelevant of $\gamma_{kn}$} (see Section \ref{sec1}). Since the convergence speed of iterative optimization algorithms is mainly affected by the step-size, the convergence speed of our algorithm is quite robust to different values of $\gamma_{kn}$. We will also show that step-sizes larger than those of \cite{Lin2006} can be utilized in our algorithm to achieve a faster convergence speed.

\subsection{Distributed Implementation of Algorithm A}
We proceed to explain how to implement \textbf{Algorithm A} in a
distributed fashion. The Lagrangian maximization problems
(\ref{eq8}) and (\ref{eq10}) can be decomposed into many independent
local power allocation subproblems for each user. Specifically, the terms of the
Lagrangian \eqref{eq7} can be reassembled as
\begin{align}
L(\vec{p},\vec{\lambda},\vec{y})=&\sum_{n=1}^{N}\left[{w_n}\log_2(1+\sum_{k
\in R(n)} p_{kn} \gamma_{kn})-\sum_{k \in R(n)}\lambda_{k}p_{kn}\right.\nonumber\\
&\left.-\sum_{k \in R(n)}\frac{c_n}{2}(p_{kn}-y_{kn})^2\right]+\sum_{k=1}^{K}\lambda_{k}P_{k}.
\end{align}
Therefore, the Lagrangian maximization problems (\ref{eq8}) and
(\ref{eq10}) can be rewritten as
\begin{equation}\label{eq12}
\max_{\vec{p}\geq 0}L(\vec{p},\vec{\lambda},\vec{y}_n)=\sum_{n=1}^{N}\max_{\vec{p}_n\geq 0}B_{n}(\vec{p}_n,\vec{\lambda},\vec{y}_n)+\sum_{k=1}^{K}\lambda_{k}P_{k},
\end{equation}
where
\begin{align}
B_{n}(\vec{p}_n,\vec{\lambda},\vec{y}_n)=&{w_n}\log_2(1+\sum_{k \in R(n)}
p_{kn} \gamma_{kn})-\sum_{k \in R(n)}\lambda_{k}p_{kn}\nonumber\\
&-\sum_{k \in R(n)}\frac{c_n}{2}(p_{kn}-y_{kn})^2.\label{eq13}
\end{align}
Therefore, problems (\ref{eq8}) and (\ref{eq10}) can be decomposed into a series of local power allocation subproblems.

In practice, the resource allocation of a user is carried out at a nearby base station. However, the antennas serving this user may belong to several base stations, as illustrated in Fig.~\ref{fig1}. Therefore, neighboring base stations need to exchange information during the iterations of \textbf{Algorithm A}. The distributed implementation procedure of \textbf{Algorithm A} is described as follows:

At the $t$th iteration of \textbf{Algorithm A}, the base station assigned to the $n$th user first utilizes the channel quality information $\{\gamma_{kn}\}_{k\in
R(n)}$ to solve a subproblem of \eqref{eq8}, given by
\begin{equation}\label{eq14}
\{p_{kn}(t)\}_{k\in R(n)}={\arg\max}_{\vec{p}_n\geq 0}B_{n}(\vec{p}_n,\vec{\lambda}(t),\vec{y}_n(t)),
\end{equation}
and forwards the power allocation solutions $\{p_{kn}(t)\}_{k\in R(n)}$ to nearby base stations controlling the antennas $k\in R(n)$. Then, the base station controlling the $k$th antenna utilizes the power allocation solutions $\{p_{kn}(t)\}_{n\in U(k)}$ to update the dual variable $\lambda_k(t+1)$ according to \eqref{eq9}, and sends $\lambda_k(t+1)$ to the base station assigned to the $n$th user. Next, the base station assigned to the $n$th user solves a subproblem of \eqref{eq10}, i.e.,
\begin{equation}
\{z_{kn}(t)\}_{k\in R(n)}={\arg\max}_{\vec{p}_n\geq 0}B_{n}(\vec{p}_n,\vec{\lambda}(t+1),\vec{y}_n(t)),
\end{equation}
and utilizes the resultant solution $\{z_{kn}(t)\}_{k\in R(n)}$ to update the auxiliary variables $y_{kn}(t+1)$ according to \eqref{eq11}. Therefore, \textbf{Algorithm A} can be implemented in a totally distributed fashion, and it only requires local exchange of the power allocation solution $p_{kn}(t)$ and the dual variable $\lambda_k(t+1)$ among neighboring base stations in each iteration. In addition, when the channel power gain $\gamma_{kn}$ changes, each user sends the updated channel power gain to its assigned base station.

\subsubsection{Solution to Local Power Allocation Subproblem \eqref{eq14}}
The optimum solution to \eqref{eq14} satisfies the following
Karush-Kuhn-Tucker (KKT) conditions \cite{Boyd2004}:
\begin{align}
\frac{\partial B_{n}}{\partial p_{kn}}=&\frac{{w_n}\gamma_{kn}}{\ln2(1+\sum_{k\in R(n)}p_{kn}\gamma_{kn})}-\lambda_{k}\nonumber\\
&-c_n(p_{kn}-y_{kn}) \left\{\begin{array}{l}
=0, ~ \textrm{if}~ p_{kn}>0;\\
\leq 0, ~ \textrm{if}~ p_{kn}=0,\end{array}\right.~~\forall ~k \in R(n).\label{eq15}
\end{align}

Define
\begin{equation}
\Omega(n)=\{k\in R(n)|p_{kn}>0\}
\end{equation}
as the set of antennas
serving the $n$th user with positive power. Hence, $p_{kn}=0$ for
all $k\in R(n)/\Omega(n)$. If $\Omega(n)$ is known, the KKT
conditions in (\ref{eq15}) indicate
\begin{align}
\frac{{w_n}\gamma_{kn}}{\ln2(1{+}\sum_{k\in
\Omega(n)}p_{kn}\gamma_{kn})}-\lambda_{k}-c_n(p_{kn}&-y_{kn})=0,\nonumber\\
&\forall~k\in\Omega(n).\label{eq16}
\end{align}
By conducting a weighted summation of the equations in (\ref{eq16}),
we obtain an equation of $\sum_{k\in
\Omega(n)}p_{kn}\gamma_{kn}$, i.e.,
\begin{align}
\frac{\sum_{k\in \Omega(n)}{w_n}\gamma_{kn}^2}{\ln2(1+\sum_{k\in
\Omega(n)}p_{kn}\gamma_{kn})}&-\!\!\sum_{k\in
\Omega(n)}\gamma_{kn}\lambda_{k}-c_n\!\!\sum_{k\in
\Omega(n)}\gamma_{kn}p_{kn}\nonumber\\
&+c_n\!\!\sum_{k\in \Omega(n)}\gamma_{kn}y_{kn}=0.\label{eq17}
\end{align}
Let us define $s_{n}\triangleq\sum_{k\in \Omega(n)}p_{kn}\gamma_{kn}$,
then (\ref{eq17}) can be reformulated as
\begin{equation}\label{eq18}
c_ns_{n}^{2}+(c_n+\mu_{n})s_{n}+\mu_{n}-\gamma_{n}=0,
\end{equation}
where $\gamma_{n}=\sum_{k\in \Omega(n)}{w_n}\gamma_{kn}^{2}/\ln2$,
$\mu_{n}=\sum_{k\in
\Omega(n)}\gamma_{kn}(\lambda_{k}-c_ny_{kn})$. The root $s_{n}$
of the quadratic equation (\ref{eq18}) is given by
\begin{equation}\label{eq19}
s_{n}=\frac{1}{2c_n}[-(c_n+\mu_{n})+\sqrt{(c_n+\mu_{n})^{2}-4c_n(\mu_{n}-\gamma_{n})}].
\end{equation}
Substituting (\ref{eq19}) into (\ref{eq16}), we obtain the optimum
solution to the subproblem (\ref{eq14}) as
\begin{equation}
{p_{kn}=}\left\{\begin{array}{l l}
y_{kn}+\frac{1}{c_n}\left[\frac{{w_n}\gamma_{kn}}{\ln2(1+s_{n})}-\lambda_{k}\right],&~ \textrm{if}~k\in \Omega(n), \\
0, &~ \textrm{if}~k\in
R(n)/\Omega(n).\end{array}\right.\label{eq20}
\end{equation}

\subsubsection{A Novel Procedure to Derive $\Omega(n)$}
Until now, the left task is to determine $\Omega(n)$ in the optimum solution
to \eqref{eq14}. Let us consider the unconstrained problem
corresponding to \eqref{eq14}, i.e.,
\begin{equation}\label{eq50}
\vec{p}_{n,0}=\arg\max_{\vec{p}_n}B_{n}(\vec{p}_n,\vec{\lambda},\vec{y}_n).
\end{equation}
Our research indicates that if $p_{kn,0}$ in the solution to
\eqref{eq50} satisfies $p_{kn,0}\leq 0$, then the solution to
\eqref{eq14} must satisfy $p_{kn}=0$ (i.e., $k\notin \Omega(n)$).
This statement is expressed in the following lemma:
\begin{lemma}\label{lem1}
Suppose that $g(x)$ is a differentiable concave function on
$x\in[0,\infty)$, and $B(\vec{p})$ is defined as
$B(\vec{p})=g(\sum_{k}p_{k}\gamma_{k})-\sum_{k}(a_{k}p_{k}^{2}+b_{k}p_{k}+c_{k})$
with $a_{k}>0$. If $\{p_{k}^{*}\}={\rm{arg}\max}_{\vec{p}_{}\geq
0}B(\vec{p})$ and
$\{p_{k,0}\}={\rm{arg}\max}_{\vec{p}_{}}B(\vec{p})$, then
$p_{k}^{*}=0$ for any $k$ satisfying $p_{k,0}\leq 0$.
\end{lemma}
\begin{proof}
See Appendix A.
\end{proof}
With Lemma \ref{lem1}, we are able to compute the optimum choice of
$\Omega(n)$. The detailed procedure is given as follows:
\begin{itemize}
\item [\textbf{P-1.}] Initialization: Set $\Omega(n)=R(n)$.

\item [\textbf{P-2.}] Compute $s_{n}$ and $p_{kn}$ according to (\ref{eq19}) and (\ref{eq20}), respectively.

\item [\textbf{P-3.}] If $p_{kn}> 0$ for all $k\in \Omega(n)$, output $\Omega(n)$ and exit; otherwise, set $\Omega_{(n)}=\{k|p_{kn}>0,k\in R(n)\}$ and return to \textbf{P-2}.
\end{itemize}

\emph{Remark 1:} Lemma \ref{lem1} allows us to rule out
all the elements $k$ with $p_{kn}\leq 0$ from $\Omega(n)$ in one
iteration. Therefore, the proposed procedure can converge much faster than the method proposed in
\cite{Lin2006},\cite{Lin2004}, which is to eliminate only one
element $k$ with the smallest negative $p_{kn}$ in each iteration.
Our method significantly reduces the number of iterations to
compute $\Omega(n)$ and does not require sorting procedure of
$\{p_{kn}\}$.

\subsection{Convergence Analysis}\label{sec1}
In this subsection, we obtain the bounds on the step-sizes $\alpha_k$ to
ensure convergence. First, we need some
notations and definitions to simplify the expressions of our
theoretical analysis. Let us consider the function
\begin{equation}\label{eq61}
f(\vec{p}) =\left\{\begin{array}{l l}\sum_{n=1}^{N} {w_n}\log_2(1+\sum_{k \in R(n)} p_{kn} \gamma_{kn}),&\textrm{if}~p_{kn}\geq 0,\\
-\infty,&\textrm{otherwise},\end{array}\right.
\end{equation}
which has incorporated the power constraint $\vec{p}\geq0$ in the
definition. {The analysis of this subsection applies to any objective function $f(\vec{p})$ in the form of
$\sum_nf_n(\sum_k p_{kn}\gamma_{kn})$ with $f_n(\cdot)$ being a
concave function.}
With \eqref{eq61}, the Lagrangian \eqref{eq7} can be rewritten as
\begin{equation}
\tilde{L}(\vec{p},\vec{\lambda},\vec{y})=f(\vec{p})-\vec{p}^{T}E^{T}\vec{\lambda}-\frac{1}{2}(\vec{p}-\vec{y})^{T}V(\vec{p}-\vec{y})+\vec{\lambda}^{T}\vec{R},
\label{eq23}
\end{equation}
where $E$ is a $K\times \sum_{j=1}^N |R_{(j)}|$ dimensional matrix
with binary elements representing the relationship between the antennas
and their transmit power variables{, i.e., if the $k$th antenna is selected to serve the $n$th user, one of the $|R(n)|$ elements on the $k$th row and the $(\sum_{j=1}^{n-1} |R_{(j)}|+1)$th to the $(\sum_{j=1}^{n} |R_{(j)}|)$th columns is 1; otherwise, all of these $|R(n)|$ elements are 0. Moreover, it satisfies
\begin{equation}\label{eq2}
\sum_{i=1}^{\sum_{j=1}^N |R_{(j)}|} E_{ki}=|U(k)|,~~\forall~i, ~~\sum_{k=1}^{K} E_{ki}=1,~~\forall~k,
\end{equation}
because the $k$th antenna serves $|U(k)|$ users and each transmit power variable belongs to only one antenna.
}$V$ is a
$\sum_{j=1}^N |R_{(j)}| \times \sum_{j=1}^N |R_{(j)}|$ diagonal
matrix with diagonal elements $c_n$ representing the parameters of
the quadratic terms.
$\vec{R}=[P_1,P_2,...,P_K]^{T}$ represents the vector of maximal
transmission power of the antennas. Therefore, the Lagrangian
maximization problems (\ref{eq8}) and (\ref{eq10}) can be expressed
as
\begin{equation}
\vec{p}(t)={\textrm{arg}\max}_{\vec{p}_{}}
\tilde{L}(\vec{p},\vec{\lambda}(t),\vec{y}(t)),
\end{equation}
and
\begin{equation}
\vec{z}(t)={\textrm{arg}\max}_{\vec{p}_{}}
\tilde{L}(\vec{p},\vec{\lambda}(t+1),\vec{y}(t)),
\end{equation}
respectively. Let $A$ be a $K \times K$ diagonal
matrix with diagonal elements $\alpha_k$ representing the step size
for dual update. Let $B$ be a $\sum_{j=1}^N |R_{(j)}|
\times \sum_{j=1}^N |R_{(j)}|$ diagonal matrix with diagonal
elements $\beta$ representing the step size for auxiliary update.
Then, the dual update (\ref{eq9}) and auxiliary update (\ref{eq10})
can be rephrased as
\begin{equation}
\vec{\lambda}(t+1)=[\vec{\lambda}(t)+A(E\vec{p}(t)-\vec{R})]^{+},
\label{eq24}
\end{equation}
and
\begin{equation}
\vec{y}(t+1)=\vec{y}(t)+B(\vec{z}(t)-\vec{y}(t)). \label{eq25}
\end{equation}

We also need to define the stationary point of \textbf{Algorithm A}.
\begin{Def} A point $(\vec{y}^{*},\vec{\lambda}^{*})$ is a
stationary point of {\bf{Algorithm A}}, if
\begin{eqnarray}
&&\vec{y}^{*}=\arg \max_{\vec{p}_{}}
\tilde{L}(\vec{p},\vec{\lambda}^{*},\vec{y}^{*}), \label{eq26}\\
&&E\vec{y}^{*}-\vec{R}\leq 0,~\vec{\lambda}^{*}\geq 0,
\label{eq27}\\
&&\vec{\lambda}^{*}\otimes(E\vec{y}^{*}-\vec{R})=0, \label{eq28}
\end{eqnarray}
where $\vec{x} \otimes \vec{y}$ represents the Hadamard
(elementwise) product of two vectors $\vec{x}$ and $\vec{y}$ with
the same dimension.
\end{Def}

Let us further consider a Lagrangian maximization problem
$\vec{p}=\arg \max_{\vec{p}_{}}
\tilde{L}(\vec{p},\vec{\lambda},\vec{y})$. The KKT conditions
suggest that there must exist a subgradient $\nabla f(\cdot)$ of
$f(\cdot)$ satisfying
\begin{eqnarray}\label{eq51}
\nabla f(\vec{p})-E^T\vec{\lambda}-V(\vec{p}-\vec{y})=0.
\end{eqnarray}
Similarly, let $(\vec{y}^{*},\vec{\lambda}^{*})$ denote a stationary
point of \textbf{Algorithm A}, then we can get from \eqref{eq26}
that
\begin{eqnarray}\label{eq52}
\nabla f(\vec{y}^*)-E^T\vec{\lambda}^*=0.
\end{eqnarray}

Now we are ready to introduce the main result of this paper in the following theorem, i.e., the {sufficient} condition for the convergence of \textbf{Algorithm A}.
\begin{Theorem}\label{thm1}
{If the objective function $f(\vec{p})$ is in the form of
$\sum_nf_n(\sum_k p_{kn}\gamma_{kn})$ with $f_n(\cdot)$ being a
concave function,} and the step-size $\alpha_k$ satisfies
\begin{eqnarray}\label{eq62}
\alpha_k\leq\frac{2\min_{\{n\in U(k)\}}c_n}{3|U(k)|},
\end{eqnarray}
where $|U(k)|$ is the number of users served by the $k$th antenna, the proposed
distributed power allocation \textbf{Algorithm A} will converge to
a stationary point $(\vec{y}^{*},\vec{\lambda}^{*})$ of the
algorithm, and $\vec{p}^{*}=\vec{y}^{*}$ is an optimum solution.
\end{Theorem}

The proof of Theorem \ref{thm1} relies on the following key result:
\begin{lemma}
\label{lem2} Let $(\vec{p}_1,\vec{\lambda}_1)$ and
$(\vec{p}_2,\vec{\lambda}_2)$ be two maximizers of the Lagrangian
\eqref{eq23} for fixed auxiliary variable $\vec{y}$, i.e.,
$\vec{p}_1=\arg\max\limits_{\vec{p}}\tilde{L}\left(\vec{p},\vec{y},\vec{\lambda}_1\right)$
and
$\vec{p}_2=\arg\max\limits_{\vec{p}}\tilde{L}\left(\vec{p},\vec{y},\vec{\lambda}_2\right)$,
and $(\vec{y}^{*},\vec{\lambda}^{*})$ is a stationary point of
\textbf{Algorithm A}, then
\begin{eqnarray}\label{eq78}
&&\left[\nabla  f\left(\vec{p}_1\right)-\nabla  f\left(\vec{y}^*\right)\right]^T\left(\vec{p}_2-\vec{y}^*\right)\nonumber\\
\leq\!\!\!\!\!\!\!\!&&\frac{1}{4}\left(\vec{\lambda}_2-\vec{\lambda}_1\right)^TEV^{-1}E^T\left(\vec{\lambda}_2-\vec{\lambda}_1\right),
\end{eqnarray}
where $\nabla f(\vec{p}_1)$ and $\nabla  f\left(\vec{y}^*\right)$
are defined in \eqref{eq51} and \eqref{eq52}.
\end{lemma}
\begin{proof}
See Appendix B.
\end{proof}
With Lemma \ref{lem2}, we are able to prove Theorem 1. The details
are relegated to Appendix C. Some remarks about Theorem 1 are provided as follows:

\emph{Remark 2:} If we choose
$\alpha_k=\frac{2\min_{\{n\in U(k)\}}c_n}{3|U(k)|}$, then the
step-size parameters $\alpha_k$ do not rely on the channel coefficients $\gamma_{kn}$. On the contrary, they are
only determined by the number of users served by the $k$th antenna,
i.e., $|U(k)|$. Since the convergence speed of iterative optimization algorithms is mainly affected by the step-size, the convergence speed of our algorithm is quite robust to different values of $\gamma_{kn}$.

\emph{Remark 3:} It is worthwhile to note that the channel fading coefficients $\gamma_{kn}$
is involved in $f(\vec{p})$ on the left hand side of \eqref{eq78}, by not in
the right hand side of \eqref{eq78}.
This is the key reason that the bound on the step-size $\alpha_k$ in
Theorem \ref{thm1} is irrelevant to $\gamma_{kn}$.

\emph{Remark 4:} In \cite[Lemma 3]{Lin2006}, the authors proved that
\begin{equation}\label{eq34}
[\nabla f(\vec{p}_1)-\nabla
f(\vec{y}^{*})]^{T}(\vec{p}_2-\vec{y}^{*})\leq
\frac{1}{2}(\vec{\lambda}_1-\vec{\lambda}_2)^{T}EV^{-1}E^{T}(\vec{\lambda}_1-\vec{\lambda}_2),
\end{equation}
for the degenerated case of $\gamma_{kn}=1$. One can see that
\eqref{eq34} is looser than the inequality \eqref{eq78} in Lemma
\ref{lem2}. Moreover, in \cite[Proposition 4]{Lin2006},
the authors only proved the convergence of their algorithm for the step-sizes ${\alpha}_k=\frac{\min_{\{n\in \mathcal
{N}\}}c_n}{2\max_{\{k\in\mathcal {K}\}} |U(k)|}$, which {are smaller than} the step-sizes of our algorithm, i.e.,  $\alpha_k=\frac{2\min_{\{n\in U(k)\}}c_n}{3|U(k)|}$.
Therefore, our algorithm can
achieve a faster speed of convergence than that of
\cite{Lin2006}. Some numerical results will be provided in the next section to illustrate this.

{\emph{Remark 5:} Theorem \ref{thm1} provides a sufficient condition for the convergence of \textbf{Algorithm A}, for all the system circumstances. According to our simulation experiences, there exist some circumstances that larger step-sizes than those of \eqref{eq62} can also obtain an optimum solution to problem (\ref{eq3}). However, it is difficult to prove that such weaker conditions ensure the convergence of \textbf{Algorithm A} uniformly for all the system circumstances.}

{\emph{Remark 6:} If the channel gains change before convergence as
in the slow fading environment, the resultant power allocation
solution may not be optimum. However, \textbf{Algorithm A} is able
to track the changes of the slow fading environment to some extent.
For example, suppose that the channel gains change after the
algorithm has reached a near optimum solution. We can still use the
dual variable and auxiliary variable of the last iteration as the
initial state of the subsequent iterations. As long as the changes
of channel gains are small, the dual variable and auxiliary variable
of the last iteration is not far from the optimum solution, and the
number of iterations for convergence is much smaller than using a
random initial state.}

\section{Numerical Simulations}
\label{sec:simulations} In this section, we present some simulation
results to demonstrate the efficiency of the proposed power
allocation algorithm. We consider a downlink DAS with 7 cells. Each
cell is equipped with 7 distributed antennas, including 1 antenna
locating at the center of the cell and 6 remote antennas distributed
near the boundary of the cell. Similar with Fig.~\ref{fig1}, the
locations of these 49 antennas form a hexagonal lattice. The minimal
distance between two neighboring antennas is $D=1000$ meters. The
users are distributed uniformly in the entire network area, with the
extra constraint that the distance from a user to a nearest antenna
is no smaller than 10 meters.
The wireless channel coefficients are composed by
three components: large-scale path loss, shadowing, and small-scale
Rayleigh fading. The path loss and shadowing are determined by the SCM model for Urban Macro environments \cite{3GPP}. Specifically, the path loss is given by $PL=34.5+35\log_{10}(d)$,
where $d$ is the distance in meters between the user and the
antenna. The shadowing component satisfies a log-normal distribution with zero mean and a standard deviation of 8 dB.
For downlink CoMP transmissions, each user is served by $|R(n)|=3$ antennas, which are selected based on large-scale channel path loss. The maximal transmission power of each antenna is assumed to be the same, i.e., $P_k=P$. {The data rate weights are chosen as $w_n=1$.} 
Two users are allowed to be scheduled on the same channel, if they are served by different antennas. The bandwidth of each receiver is 1MHz, and the noise figure of each receiver is 5 dB.
The conservative noise-plus-interference power $\sigma_{n,peak}^2$ is chosen to be 5 dB greater than the noise power. Therefore, the noise-plus-interference power at each receiver is $\sigma_{n,peak}^2=-174+60+5 + 5 =-104$ dBm. We utilize $\tilde{\mathbf{C}}_{n}$ in \eqref{eq1} to formulate the power allocation problem \eqref{eq3}. After the power allocation solution is derived, we substitute it into the original rate function $\mathbf{C}_{n}$ in \eqref{eq0} to compute the achievable data rate. All the simulation results are obtained by averaging over 1000 system realizations.

We compare our proposed power allocation strategy for problem \eqref{eq3} with the following 2 reference strategies: The first strategy considers the
optimal power allocation for downlink CoMP transmissions with
\emph{no} interference, which provides a performance upper bound of the
practical scenarios with interference. The second one is a simple equal power allocation (EPA) strategy, where each antenna allocates its transmission power equally to serve its users.

Figure~\ref{fig2} illustrates the simulation results of per-user throughput
versus transmission power $P$ for different power allocation strategies, where each cell has 10 users.

\begin{figure}[t]
    \centering
    \scalebox{0.6}{\includegraphics*[86,261][506,577]{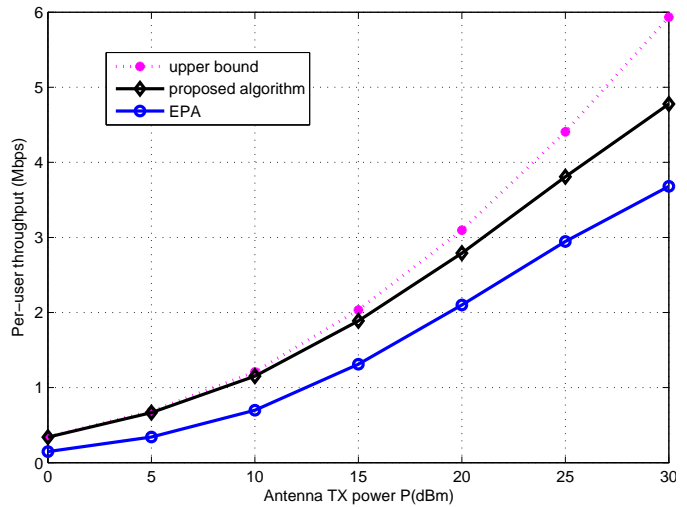}}
    \caption{Simulation results of per-user throughput versus transmission power $P$ for $N=70$.} \label{fig2}
\end{figure}

Figure~\ref{fig3} presents the simulation results of per-user throughput
versus the number of users per cell $N/7$, where the transmission
power $P=20$ dBm. One can observe that the proposed power allocation
strategy has a small gap from the performance upper bound,
especially when the transmission power $P$ is small. However, the
simple equal power allocation scheme has a lower throughput.
The performance of equal power allocation is poor, because the
wireless links from different antennas to one user have quite
different channel quality. The base station should spend more power
on the strong wireless links, instead of using the same power for
different wireless links. Through careful user scheduling and
setting reasonable threshold of noise amplification, the proposed
algorithm can achieve performance approaching to that of the ideal
non-interference scenario. Therefore, the proposed power allocation
strategy plays an essential role to realize the benefits of downlink
CoMP transmissions in DAS.

\begin{figure}[t]
    \centering
    \scalebox{0.6}{\includegraphics*[86,261][510,578]{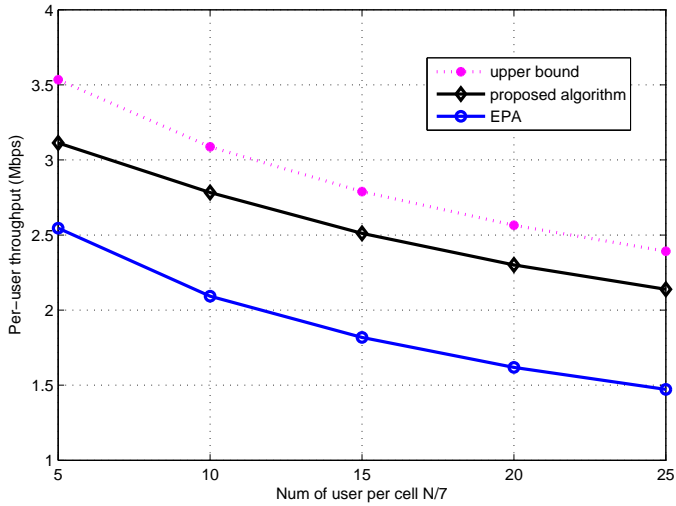}}
    \caption{Simulation results of per-user throughput versus the number of users per cell $N/7$ for $P=20$dBm.} \label{fig3}
\end{figure}

Figure~\ref{fig4} illustrates the evolutions of the dual optimality gap of
the proposed power allocation algorithm and the distributed
optimization algorithm of \cite{Lin2006} for $N=175$ and $P=30$dBm,
where the dual optimality gap is given by
$L(\vec{p}(t),\vec{\lambda}(t),\vec{y}(t))-f(\vec{y}^*)$. The
parameters of our distributed power allocation algorithm are chosen
as $c_n=3$, $\alpha_k=\frac{2\min_{\{n\in U(k)\}}c_n}{3|U(k)|}$, and
$\beta=1$. The parameters of the algorithm of \cite{Lin2006} are
given by $c_n=3$, $\alpha_k=\frac{\min_{\{n\in\mathcal
{N}\}}c_n}{2\max_{\{k\in\mathcal {K}\}}|U(k)|}$ (see \emph{Remark
4}), and $\beta=1$. Since the step-sizes of our proposed power
allocation algorithm are larger than the reference algorithm in
\cite{Lin2006}, our algorithm exhibits a faster convergence speed.
We note that this is the convergence speed when the algorithm
is cold started. In practice, since the channel condition varies
slowly, the power allocation solution from the previous run of the
algorithm is an excellent initial state for warm-starting the
algorithm. By this, the algorithm generally converges much faster.

\begin{figure}[t]
    \centering
    \scalebox{0.6}{\includegraphics*[86,261][510,578]{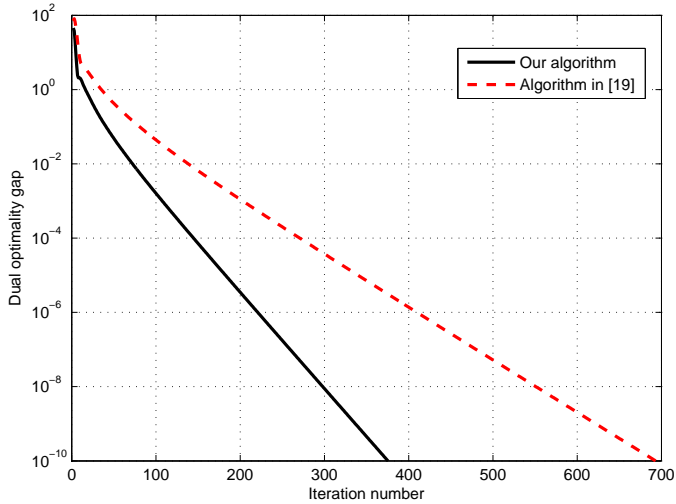}}
    \caption{Convergence of our proposed algorithm and the algorithm in \cite{Lin2006} for $N=175$ and $P=25$dBm.} \label{fig4}
\end{figure}

\section{Conclusions}
\label{sec:conclusion} We have proposed a distributed power
allocation algorithm for downlink CoMP transmissions in DAS. We
considered an approximate power allocation problem with a
non-strictly concave objective function, which makes traditional
duality-based optimization techniques not applicable for this
problem. We have resolved this non-strict concavity issue by adding
some quadratic terms to make the objective function strictly
concave, and developed a distributed algorithm to solve the power
allocation problem. A key merit of this algorithm is that its
convergence speed is robust to different values of the channel
coefficients. Its implementation only requires local information
exchange among neighboring base stations serving the same
user. The convergence and optimality of this algorithm has been
established rigorously. Our simulation results have revealed that
significant throughput improvements can be realized by this power
allocation algorithm.

\appendices {\setcounter{equation}{0}
\renewcommand{\theequation}{A.\arabic{equation}}
\section{Proof of Lemma 1}

Since $g(x)$ is a concave function of $x$,
$B(\vec{p})=g(\sum_{k}p_{k}\gamma_{k})-\sum_{k}(a_{k}p_{k}^{2}+b_{k}p_{k}+c_{k})$
is also concave with respect to $\vec{p}$. The KKT conditions
indicate
\begin{equation}
\left.\frac{\partial B}{\partial p_k}\right|_{p_k = p_{k,0}}=0,
\label{eq35}
\end{equation}
and
\begin{equation} \label{eq36}
\left.\frac{\partial B}{\partial p_k}\right|_{p_k = p^*_{k}}\left\{\begin{aligned}
         =0 &~~~{\rm if}~p^{*}_k>0, \\
                  \leq 0
                  &~~~{\rm if}~p^{*}_k=0.
                          \end{aligned} \right.
                          \end{equation}

By taking the weighted summation of the partial derivations
$\frac{\partial B}{\partial p_k}$, we obtain
\begin{equation}
\sum_k \frac{\gamma_k}{a_k}\frac{\partial B}{\partial
p_k}=g'(\sum_{k}p_k
\gamma_k)\sum_{k}\frac{\gamma_k^2}{a_k}-\sum_{k}(2
p_k\gamma_k+\frac{b_k \gamma_k}{a_k}). \label{eq37}
\end{equation}
Let $s=\sum_{k}p_{k,0} \gamma_k$, \eqref{eq35} and \eqref{eq37}
imply
\begin{equation}
g'(s)\sum_{k}\frac{\gamma_k^2}{a_k}-2s-\sum_{k} \frac{b_k
\gamma_k}{a_k}=0. \label{eq38}
\end{equation}
For $s^{*}=\sum_{k}p_k^{*} \gamma_k$, \eqref{eq36} and \eqref{eq37}
suggest
\begin{equation}
g'(s^{*})\sum_{k}\frac{\gamma_k^2}{a_k}-2s^{*}-\sum_{k} \frac{b_k
\gamma_k}{a_k}\leq 0. \label{eq39}
\end{equation}
Comparing (\ref{eq38}) and (\ref{eq39}), we derive that
\begin{equation}
[g'(s)-g'(s^{*})]\sum_{k}\frac{\gamma_k^2}{a_k}-2(s-s^{*})\geq 0.
\label{eq40}
\end{equation}
If $s-s^{*}\neq 0$, then
\begin{equation}
(s-s^{*})\left[\frac{g'(s)-g'(s^{*})}{s-s^{*}}\sum_{k}\frac{\gamma_k^2}{a_k}-2\right]\geq
0. \label{eq41}
\end{equation}
Since $g(x)$ is a concave function,
$\frac{g'(s)-g'(s^{*})}{s-s^{*}}<0$. Further, by the positivity of
$a_k$, we have $s-s^{*}\leq 0$.

Suppose there exists some $k$ such that $p_{k,0}\leq 0$ and
$p_k^{*}> 0$. Then, \eqref{eq35} and \eqref{eq36} imply
\begin{equation}
\left.\frac{\partial B}{\partial p_k}\right|_{p_k=p_{k,0}}=
\left.\frac{\partial B}{\partial p_k}\right|_{p_k=p_{k}^{*}}=0,
\end{equation}
which further suggests
\begin{equation}
p_{k,0}=\frac{1}{2a_k}\left[\gamma_k g'(s)-b_k\right],~
p_k^{*}=\frac{1}{2a_k}\left[\gamma_k g'(s^{*})-b_k\right].
\end{equation}
Since $s-s^{*}\leq 0$ and $g(x)$ is a concave function, we have
$g'(s)\geq g'(s^{*})$. Therefore, $p_k^{*}\leq p_{k,0}$, which
contradicts with the assumption of $p_{k,0}\leq 0$ and $p^{*}_k>0$.
Therefore, if $p_{k,0}\leq 0$, $p^{*}_k=0$.

\section{Proof of Lemma \ref{lem2}}
{\setcounter{equation}{0}
\renewcommand{\theequation}{B.\arabic{equation}}

We need to use the fact that $f(\vec{p})$ is in the form of
$\sum_nf_n(\sum_k p_{kn}\gamma_{kn})$, where $f_n(\cdot)$ is a
concave function.
Equation \eqref{eq51} can be also written as
\begin{align}
&\nabla f_n(\sum_{k\in R(n)}
p_{kn}\gamma_{kn})\gamma_{kn}{-}\lambda_{k}{-}c_n(p_{kn}{-}y_{kn})=0,~\forall
k\in R(n),\label{eq43}
\end{align}
where $\nabla f_n(\cdot)$ is the subgradient of $f_n(\cdot)$. By
conducting a weighted summation of the equations in \eqref{eq43}, we
obtain
\begin{align}
\nabla f_n(&\sum_{k\in R(n)} p_{kn}\gamma_{kn})\sum_{k\in
R(n)}\gamma_{kn}^2-\sum_{k\in
R(n)}\lambda_{k}\gamma_{kn}\nonumber\\
&-c_n\sum_{k\in R(n)}
p_{kn}\gamma_{kn}+c_n\sum_{k\in R(n)}
y_{kn}\gamma_{kn}=0.\label{eq57}
\end{align}
Let us define $(i=1,2)$
\begin{eqnarray}
\hspace{-1cm}&&{a}_{n,i}=\nabla f_n(\sum_{k\in R(n)} p_{kn,i}\gamma_{kn}) - \nabla f_n(\sum_{k\in R(n)} y^*_{kn}\gamma_{kn}),\\
\hspace{-1cm}&&{b}_{n,i}=\sum_k p_{kn,i}\gamma_{kn}-\sum_k
y_{kn}^{*}\gamma_{kn}.
\end{eqnarray}

Then, \eqref{eq57} indicates
\begin{align}
&(a_{n,1}-a_{n,2})\sum_{k\in
R(n)}\gamma_{kn}^2-c_n(b_{n,1}-b_{n,2})\nonumber\\
=&\sum_{k\in
R(n)}(\lambda_{k,1}-\lambda_{k,2})\gamma_{kn}, ~k\in R(n).\label{eq53}
\end{align}

Then, the formula on the left hand side of \eqref{eq78} can be
written as
\begin{align}\label{eq55}
\left[\nabla f(\vec{p}_1)-\nabla
f(\vec{y}^{*})\right]^{T}(\vec{p}_{2}-\vec{y}^{*})=\sum_{n}{a}_{n,1}{b}_{n,2}.
\end{align}
Since $f_n(\cdot)$ is a concave function, we obtain
\begin{eqnarray}
{a}_{n,1}{b}_{n,1}\leq 0,~~{a}_{n,2}{b}_{n,2}\leq 0.\label{eq54}
\end{eqnarray}
Hence, $-\frac{c_n{b}_{n,1}}{{a}_{n,1}\sum_{k\in
R(n)}\gamma_{kn}^2}\geq 0$.

Then
\begin{align}
&~{a}_{n,1}{b}_{n,2}\left(1-\frac{c_n{b}_{n,1}}{{a}_{n,1}\sum_{k\in
R(n)}\gamma_{kn}^2}\right)\nonumber\\
=&\left({a}_{n,1}-\frac{c_n{b}_{n,1}}{\sum_{k\in R(n)}\gamma_{kn}^2}\right){b}_{n,2}\nonumber\\
=&\left[{a}_{n,2}-\frac{c_n{b}_{n,2}}{\sum_{k\in R(n)}\gamma_{kn}^2}\right.\nonumber\\
&\left.\quad\quad\quad\quad+\frac{\sum_{k\in R(n)}(\lambda_{k,1}-\lambda_{k,2})\gamma_{kn}}{\sum_{k\in R(n)}\gamma_{kn}^2}\right]{b}_{n,2} ~~~(\textrm{by } \eqref{eq53})\nonumber\\
\leq&\frac{-c_n{b}_{n,2}^2{+}{b}_{n,2}\sum_{k\in R(n)}(\lambda_{k,1}{-}\lambda_{k,2})\gamma_{kn}}{\sum_{k\in R(n)}\gamma_{kn}^2}~(\textrm{by } a_{n,2}b_{n,2}\leq 0)\nonumber\\
\leq & \frac{\left[\sum_{k\in R(n)}(\lambda_{k,1}-\lambda_{k,2})\gamma_{kn}\right]^2}{4c_n\sum_{k\in R(n)}\gamma_{kn}^2} ~(\textrm{by completing the square})\nonumber\\
\leq & \frac{1}{4c_n}\sum_{k\in
R(n)}(\lambda_{k,1}-\lambda_{k,2})^2~(\textrm{by
Cauchy-Schwarz}).
\end{align}
Therefore
\begin{align}\label{eq47}
{a}_{n,1}{b}_{n,2}\leq \frac{1}{4c_n}\sum_{k\in
R(n)}(\lambda_{k,1}-\lambda_{k,2})^{2}.
\end{align}
The statement of Lemma \ref{lem2} follows by substituting
(\ref{eq47}) into (\ref{eq55}).

\section{Proof of Theorem \ref{thm1}}
{\setcounter{equation}{0}
\renewcommand{\theequation}{C.\arabic{equation}}
Let us define the norm of dual and auxiliary variables:
\begin{equation}
\|
\vec{\lambda}\|_{A}=\vec{\lambda}^{T}A^{-1}\vec{\lambda},~\|\vec{y}\|_{V }=\vec{y}^{T} V\vec{y},~~\|\vec{y}\|_{B V}=\vec{y}^{T} B^{-1} V
\vec{y}.\nonumber
\end{equation}
Suppose that $(\vec{y}^{*},\vec{\lambda}^{*})$ is a stationary point
of {\textbf{Algorithm A}}, we will show that the Lyapunov function
\begin{equation}
v(\vec{y}(t),\vec{\lambda}(t))=\|\vec{\lambda}(t)-\vec{\lambda}^{*}\|_{A}+\|\vec{y}(t)-\vec{y}^{*}\|_{B
V}
\end{equation}
is non-increasing in iteration number $t$.

In \cite{Lin2006}, it was shown that
\begin{align}
&~v(\vec{y}(t+1),\vec{\lambda}(t+1))-v(\vec{y}(t),\vec{\lambda}(t))\nonumber\\
\leq&-\|\vec{\lambda}(t+1)-\vec{\lambda}(t)\|_{A}\nonumber\\
&+(\vec{\lambda}(t+1)-\vec{\lambda}(t))^{T}EV^{-1}E^{T}(\vec{\lambda}(t+1)-\vec{\lambda}(t))\nonumber\\
&-\|\vec{y}(t)-\vec{p}(t)\|_{V} \nonumber\\
&+2[\nabla f(\vec{z}(t))-\nabla
f(\vec{y}^{*})]^{T}(\vec{p}(t)-\vec{y}^{*}).\label{eq30}
\end{align}

Invoking Lemma \ref{lem2}, we have
\begin{align}
& [\nabla f(\vec{z}(t))-\nabla f(\vec{y}^{*})]^{T}(\vec{p}(t)-\vec{y}^{*}) \nonumber\\
\leq & \frac{1}{4}(\vec{\lambda}(t+1)-\vec{\lambda}(t))^{T}EV^{-1}E^{T}(\vec{\lambda}(t+1)-\vec{\lambda}(t)).
\label{eq31}
\end{align}

Substituting (\ref{eq31}) into (\ref{eq30}), we obtain
\begin{align}
&v(\vec{y}(t+1),\vec{\lambda}(t+1))-v(\vec{y}(t),\vec{\lambda}(t)) \nonumber\\
\leq
&-(\vec{\lambda}(t+1)-\vec{\lambda}(t))^{T}C(\vec{\lambda}(t+1)-\vec{\lambda}(t))-\|\vec{y}(t)-\vec{p}(t)\|_{V},
\label{eq32}
\end{align}
where $C=A^{-1}-\frac{3}{2}EV^{-1}E^{T}$. If $C$ is non-negative
definite, then the Lyapunov function
$v(\vec{y}(t),\vec{\lambda}(t))$ is non-increasing in iteration
number $t$. Then, we can prove that
$(\vec{y}^{*},\vec{\lambda}^{*})$ is a stationary point of
\textbf{Algorithm A} by using the standard Lyapunov drift arguments
in \cite[Prop. 4]{Lin2006}. According to the standard duality
theory, if $(\vec{y}^{*},\vec{\lambda}^{*})$ is a stationary point
of \textbf{Algorithm A}, then $\vec{p}^{*}=\vec{y}^{*}$ provides a
solution to (\ref{eq3}).

Finally, we need to show that $C$ is a non-negative definite matrix,
if \eqref{eq62} is true. Let $\vec{x}$ be any vector of $K$
dimensions, according to \eqref{eq2}, we have
\begin{eqnarray}\label{eq63}
&&~~~\vec{x}^TA^{-1}\vec{x}-\frac{3}{2}\vec{x}^TEV^{-1}E^{T}\vec{x}\nonumber\\
&&=\sum_{k=1}^K a_k^{-1}x_k^2-\frac{3}{2}\sum_{n=1}^N\sum_{k\in R(n)}c_n^{-1} x_k^2\nonumber\\
&&=\sum_{k=1}^K a_k^{-1}x_k^2-\frac{3}{2}\sum_{k=1}^K\sum_{n\in U(k)} c_n^{-1} x_k^2\nonumber\\
&&=\sum_{k=1}^K \left(a_k^{-1}-\frac{3}{2}\sum_{n\in U(k)}
c_n^{-1}\right) x_k^2.
\end{eqnarray}
By \eqref{eq62}, we can obtain
\begin{eqnarray}
a_k\leq \frac{2\min_{n\in U(k)}c_n}{3|U(k)|}.
\end{eqnarray}
This further suggests
\begin{eqnarray}\label{eq64}
a_k^{-1}\geq \frac{3}{2}|U(k)|\max_{n\in U(k)}c_n^{-1}\geq \frac{3}{2}\sum_{n\in U(k)}c_n^{-1}.
\end{eqnarray}
Substituting \eqref{eq64} into \eqref{eq63}, we obtain that
$\vec{x}^TA^{-1}\vec{x}-\frac{3}{2}\vec{x}^TEV^{-1}E^{T}\vec{x}\geq0$
for any $\vec{x}$. Therefore, $C$ is a non-negative definite matrix.


\vspace{-40pt}
\begin{IEEEbiography}[\resizebox{1 in}{!}{\includegraphics{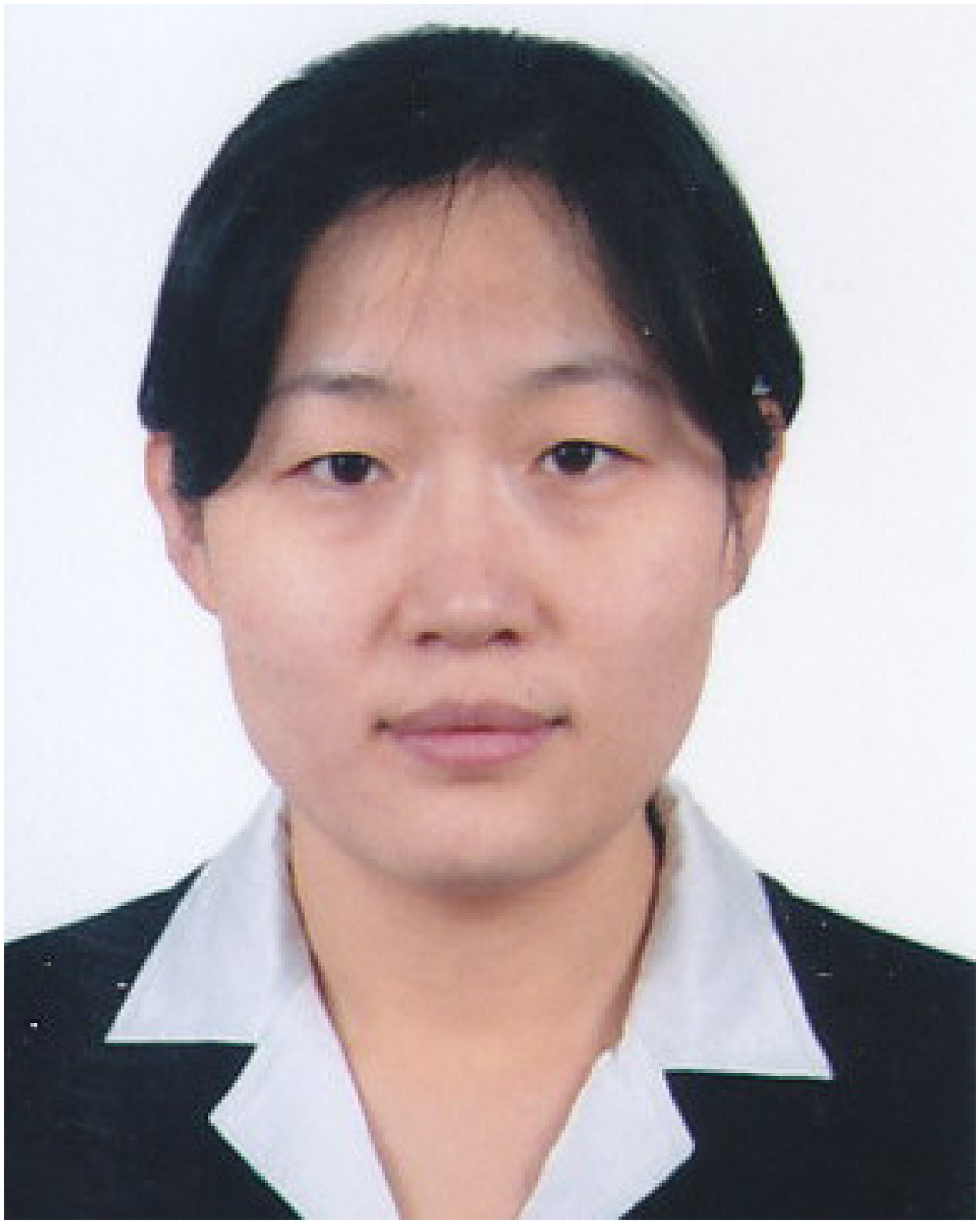}}]
{Xiujun Zhang} received B.E. and M.S. degrees in electronic engineering from Tsinghua University, Beijing, China, in 2001 and 2004, respectively. She is currently pursuing the Ph.D. degree with the Wireless and Mobile Communication Technology R\&D Center, Research Institute of Information Technology, Tsinghua University. Her research interests are in the area of signal processing and wireless communications.
\end{IEEEbiography}

\vspace{-40pt}
\begin{IEEEbiography}[\resizebox{1 in}{!}{\includegraphics{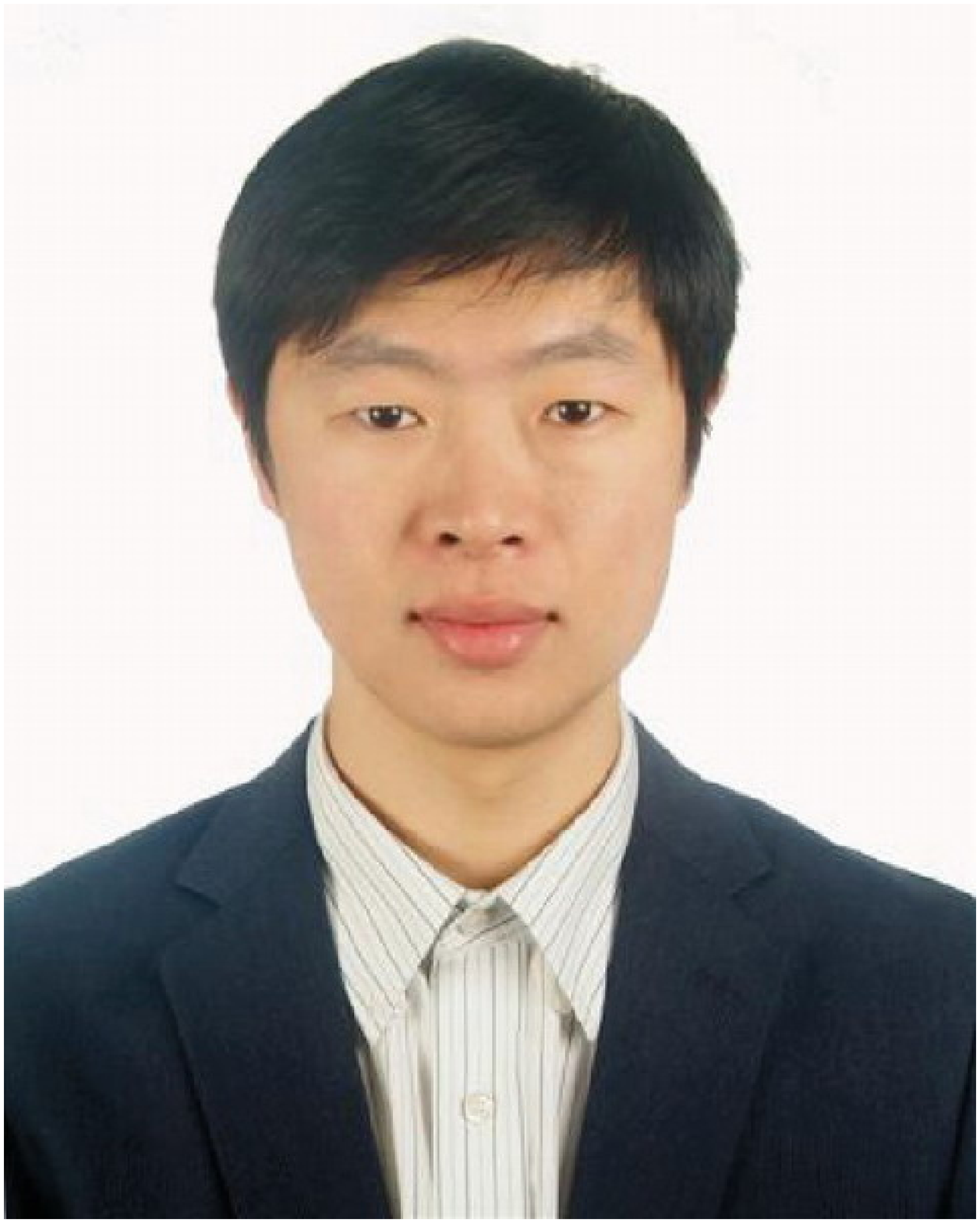}}]
{Yin Sun} (S'08-M'11) received the B. Eng. Degree and Ph.D. degree in electrical engineering from Tsinghua University, Beijing, China, in 2006 and 2011, respectively.

He is currently a Post-doctoral Researcher at the Ohio State University. His research interests include probability theory, optimization, information theory and wireless communications.

Dr. Sun received the Tsinghua University Outstanding Doctoral Dissertation Award in 2011.
\end{IEEEbiography}

\vspace{-20pt}
\begin{IEEEbiography}[\resizebox{1 in}{!}{\includegraphics{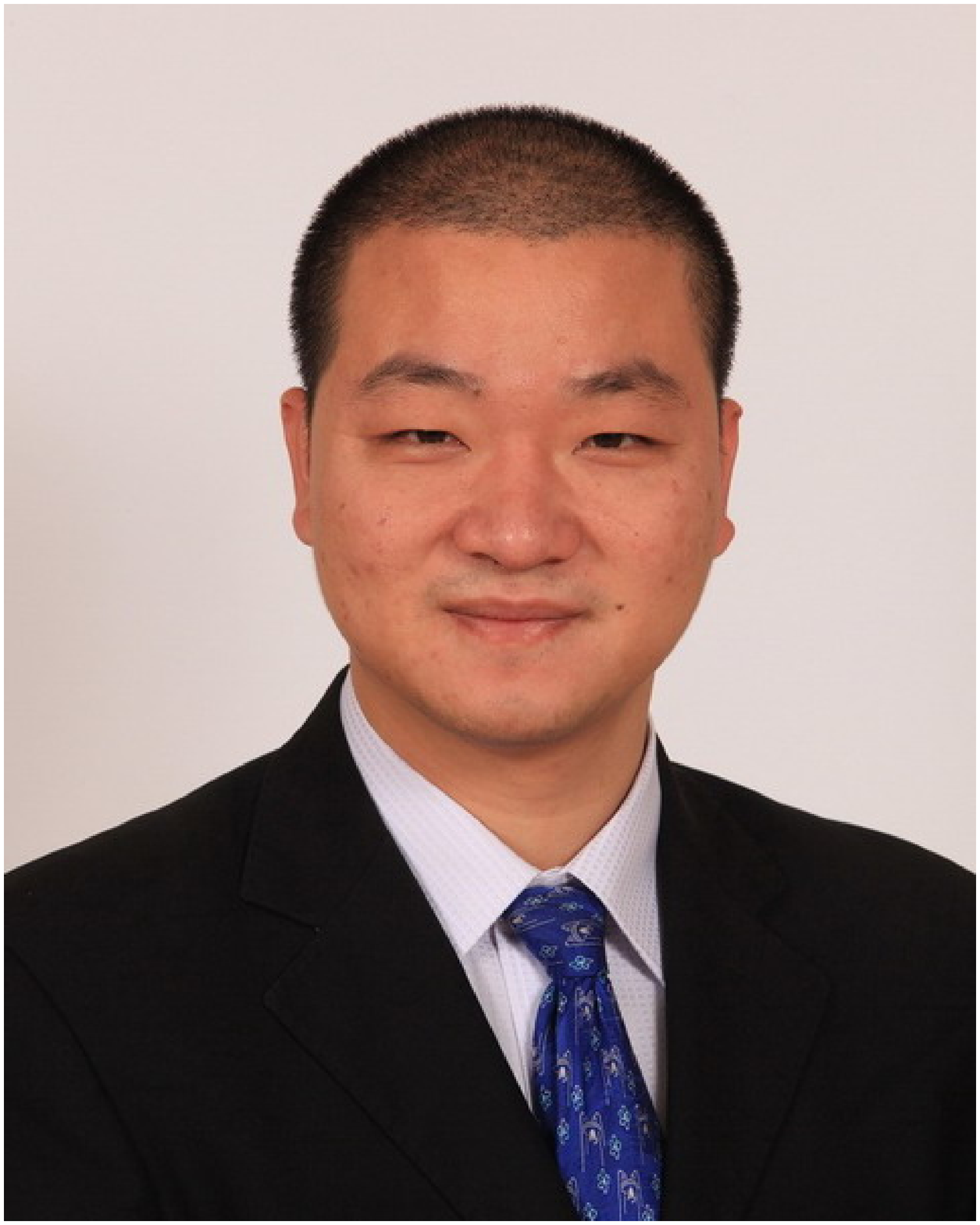}}]
{Xiang Chen} (S'02-M'07) received the B.E. and Ph.D. degrees both from the Department of Electronic Engineering, Tsinghua University, Beijing, China, in 2002 and 2008, respectively.

From July 2008 to July 2012, he was with the Wireless and Mobile Communication Technology R\&D Center (Wireless Center), Research Institute of Information Technology in Tsinghua University. Since August 2012, he serves as an assistant researcher at Aerosapce Center, School of Aerospace, Tsinghua University, Beijing, China. During July 2005 and August 2005, he was an internship student at Audio Signal Group of Multimedia laboratories, NTT DoCoMo R\&D, In YRP, Japan. During September 2006 and April 2007, he was a visiting research student at Wireless Communications \& Signal Processing (WCSP) Lab of National Tsing Hua University, Hsinchu, Taiwan. Dr. Chen's research interests mainly focus on statistical signal processing, digital signal processing, software radio, and wireless communications.
\end{IEEEbiography}

\vspace{-40pt}
\begin{IEEEbiography}[\resizebox{1 in}{!}{\includegraphics{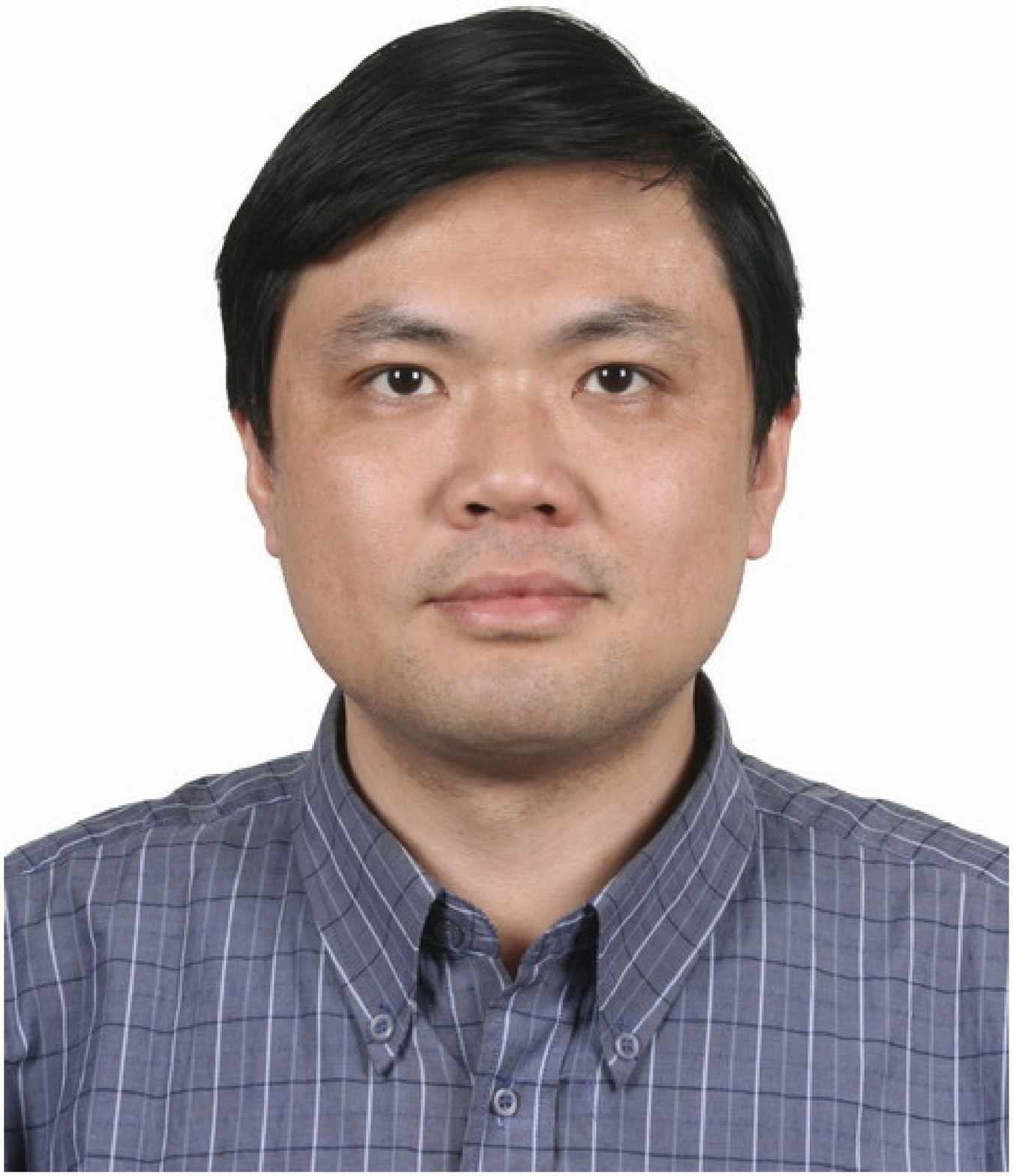}}] {Shidong Zhou} (M'98) received the B.S. and M.S. degrees in wireless communications were received from Southeast University, Nanjing, China, in 1991 and 1994, respectively, and the Ph.D. degree in communication and information systems from Tsinghua University, Beijing, China, in 1998.

From 1999 to 2001, he was in charge of several projects in the China 3G Mobile Communication R\&D Project. He is currently a Professor at Tsinghua University. His research interests are in the area of wireless and mobile communications.
\end{IEEEbiography}

\vspace{-20pt}
\begin{IEEEbiography}[\resizebox{1 in}{!}{\includegraphics{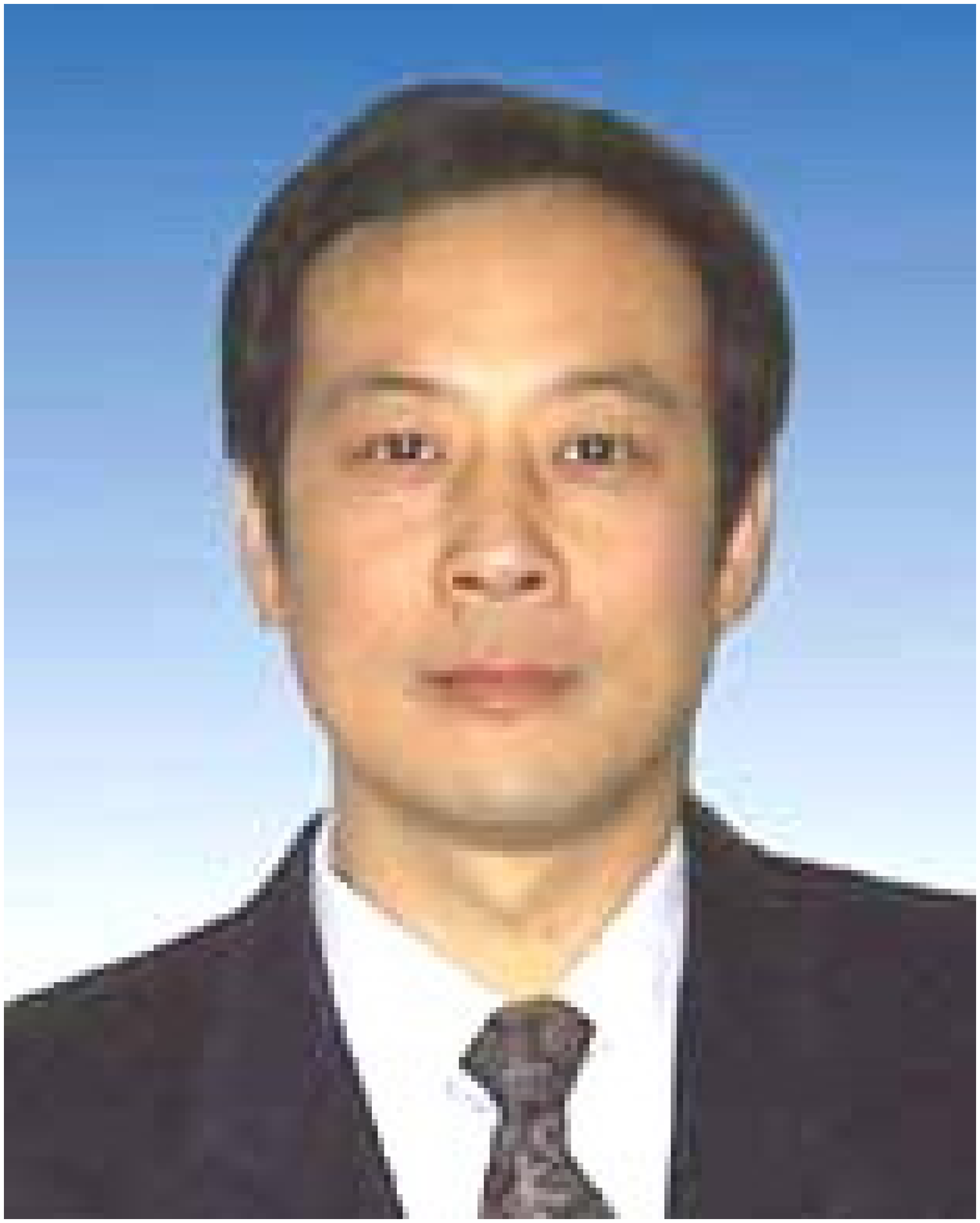}}]
{Jing Wang} (M'99) received the B.S. and M.S. degrees in electronic engineering from Tsinghua University, Beijing, China, in 1983 and 1986, respectively.

He has been on the faculty of Tsinghua University since 1986. He currently is a Professor and the Vice-Dean of the Tsinghua National Laboratory for Information Science and Technology. His research interests are in the area of wireless digital communications, including modulation, channel coding, multiuser detection, and 2D RAKE receivers. He has published more than 100 conference and journal papers. He is a member of the Technical Group of China 3G Mobile Communication R\&D Project, and he serves as an expert of communication technology in the National 863 Program.

Mr. Wang is also a member of the Radio Communication Committee of Chinese Institute of Communications and a senior member of the Chinese Institute of Electronics.
\end{IEEEbiography}

\vspace{-9cm} \vspace{-20pt}
\begin{IEEEbiography}[\resizebox{1 in}{!}{\includegraphics{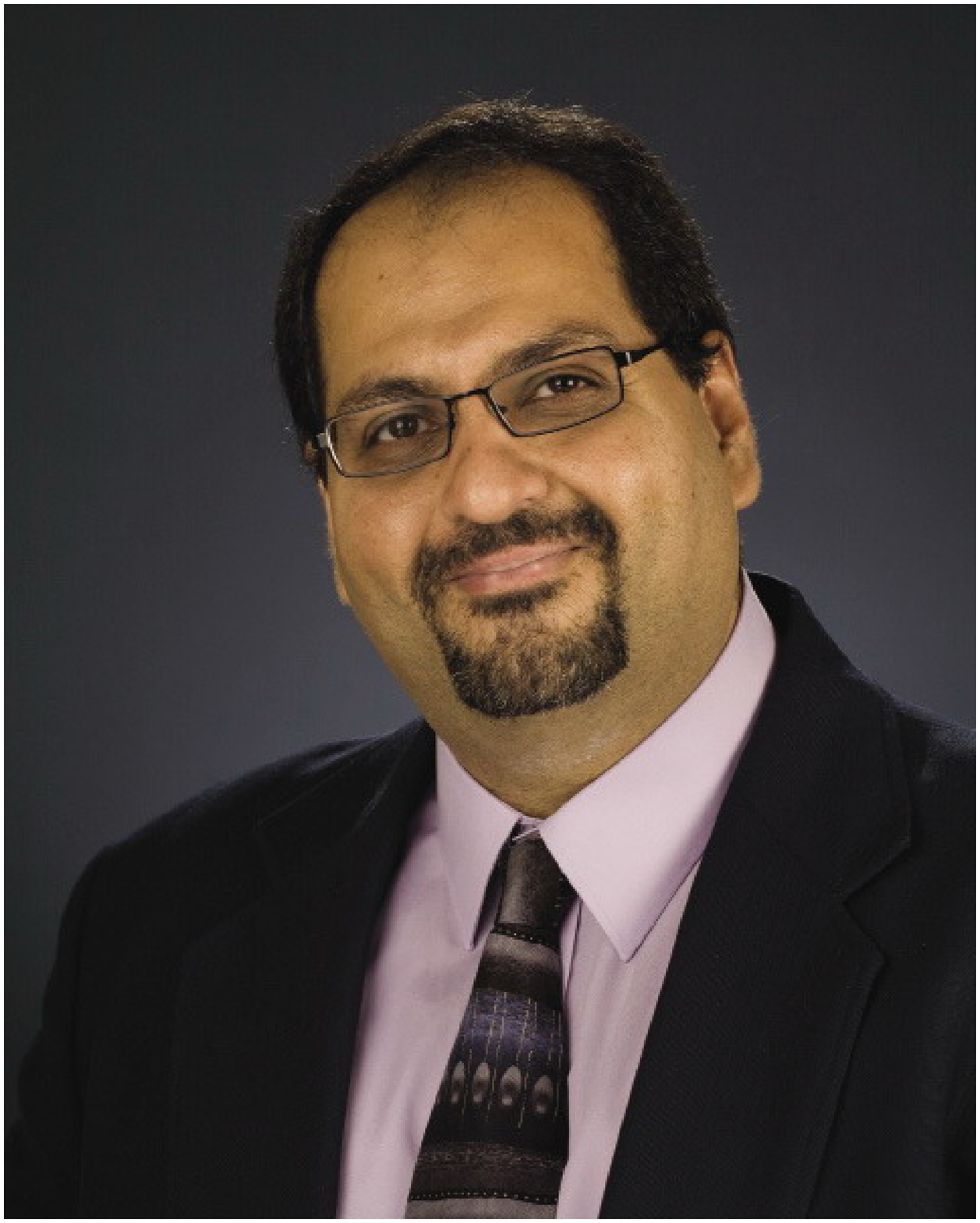}}]
{Ness B. Shroff} (F'07) received his Ph.D. degree from Columbia University, NY, in 1994 and joined Purdue University as an Assistant Professor. At Purdue, he became Professor of the School of Electrical and Computer Engineering in 2003 and director of CWSA in 2004, a university-wide center on wireless systems and applications. In July 2007, he joined The Ohio State University as the Ohio Eminent Scholar of Networking and Communications, a chaired Professor of ECE and CSE. He is also a guest chaired professor of Wireless Communications and Networking in the department of Electronic Engineering at Tsinghua University. His research interests span the areas of wireless and wireline communication networks. He is especially interested in fundamental problems in the design, performance, pricing, and security of these networks.

Dr. Shroff is a past editor for IEEE\/ACM Trans. on Networking and the IEEE Communications Letters and current editor of the Computer Networks Journal. He has served as the technical program co-chair and general co-chair of several major conferences and workshops, such as the IEEE INFOCOM 2003, ACM Mobihoc 2008, IEEE CCW 1999, and WICON 2008. He was also a co-organizer of the NSF workshop on Fundamental Research in Networking (2003) and the NSF Workshop on the Future of Wireless Networks (2009). Dr. Shroff is a fellow of the IEEE. He received the IEEE INFOCOM 2008 best paper award, the IEEE INFOCOM 2006 best paper award, the IEEE IWQoS 2006 best student paper award, the 2005 best paper of the year award for the Journal of Communications and Networking, the 2003 best paper of the year award for Computer Networks, and the NSF CAREER award in 1996 (his INFOCOM 2005 paper was also selected as one of two runner-up papers for the best paper award).
\end{IEEEbiography}


\begin{thebibliography}{1}
\bibitem{Li2011}
H. Li, J. Hajipour, A. Attar, and V.C.M. Leung, ``Efficient HetNet implementation using broadband wireless access with fiber-connected massively distributed antennas architecture,'' {\sl IEEE Wireless Commun.}, vol. 18, no. 3, pp.72-78, Jun. 2011.

\bibitem{Choi2007}
W. Choi, and J. G. Andrews, ``Downlink performance and capacity of
distributed antenna systems in a multicell environment,'' {\sl IEEE Trans. Wireless Commun.}, vol. 6, no. 1, pp. 69-73, Jan. 2007.

\bibitem{Park2009}
J. Park, E. Song, and W. Sung, ``Capacity analysis for distributed
antenna systems using cooperative transmission schemes in fading
channels,'' {\sl IEEE Trans. Wireless Commun.}, vol. 8, no. 2,
pp. 586-592, Feb. 2009.

\bibitem{Sawahashi2010}
M. Sawahashi, Y. Kishiyama, A. Morimoto, D. Nishikawa, and M. Tanno,
``Coordinated multipoint transmission/reception techniques for
LTE-Advanced,'' \emph{IEEE Wireless Commun.}, vol. 17, no. 3, pp.
26-34, Jun. 2010.

\bibitem{Venturino2009}
L. Venturino, N. Prasad, and Xiaodong Wang, ``Coordinated scheduling
and power allocation in downlink multicell OFDMA networks,'' {\sl
IEEE Trans. Veh. Technol.}, vol. 58, no. 6, pp. 2835-2848, Jul.
2009.

\bibitem{Papandreou2008}
N. Papandreou, and T. Antonakopoulos, ``Bit and power allocation in
constrained multicarrier systems: The single-user case,'' {\sl
EURASIP Journal on Advances in Signal Processing}, vol. 2008,
Article ID 643081, 14 pages, 2008.

\bibitem{Gesbert2010}
D. Gesbert, S. Hanly, H. Huang, S. Shamai, O. Simeone, and Y. Wei, ``Multi-cell MIMO cooperative networks: A new look at interference,'' {\sl IEEE J. Sel. Areas
Commun.}, vol. 28, no. 9, pp. 1380-1408, Dec. 2010.

\bibitem{Luo2011}
B. Luo, Q. Cui, H. Wang, and X. Tao, ``Optimal joint water-filling
for coordinated transmission over frequency-selective fading
channels,'' {\sl IEEE Commun. Letters}, vol. 15, no. 2, pp. 190-192,
Feb. 2011.

\bibitem{Yu2002}
W. Yu, G. Ginis, and J. Cioffi, ``Distributed multiuser power
control for digital subscriber lines,'' {\sl IEEE J. Sel. Areas
Commun.}, vol. 20, no. 5, pp. 1105-1115, Jun. 2002.

\bibitem{Yu2004}
W. Yu, W. Rhee, S. Boyd, and J. Cioffi, ``Iterative water-filling
for Gaussian vector multiple access channels,'' {\sl IEEE Trans.
Inf. Theory}, vol. 50, no.1, pp. 145-152, Jan. 2004.

\bibitem{Huang2006}
J. Huang, R. A. Berry, and M. L. Honig, ``Distributed interference
compensation for wireless networks,'' {\sl IEEE J. Sel. Areas
Commun.}, vol. 24, no. 5, pp. 1074-1084, May 2006.

\bibitem{Gesbert2007}
D. Gesbert, S. Kiani, A. Gjendemsj, and G. E. ien, ``Adaptation,
coordination, and distributed resource allocation in
interference-limited wireless networks,'' {\sl Proc. IEEE}, vol. 95,
no. 12, pp. 2393-2409, 2007.

\bibitem{Pang2008}
J.-S. Pang, G. Scutari, F. Facchinei, and C. Wang, ``Distributed
power allocation with rate constraints in Gaussian parallel
interference channels,'' {\sl IEEE Trans. Inf. Theory}, vol. 54,
no.8, pp. 3471-3489, 2008.

\bibitem{Xiao04}
L.~Xiao, M.~Johansson, and S.~P. Boyd, ``Simultaneous routing and
resource
  allocation via dual decomposition,'' \emph{IEEE Trans. Commun.}, vol.~52,
  no.~7, pp. 1136--1144, Jul. 2004.

\bibitem{LinJsac06}
X.~Lin, N.~Shroff, and R.~Srikant, ``A tutorial on cross-layer
optimization in
  wireless networks,'' \emph{IEEE J. Sel. Areas Commun.}, vol.~24, no.~8, pp.
  1452--1463, Aug. 2006.

\bibitem{Chiang07ProcIEEE}
M.~Chiang, S.~Low, A.~Calderbank, and J.~Doyle, ``Layering as
optimization
  decomposition: A mathematical theory of network architectures,'' \emph{Proc.
  IEEE}, vol.~95, no.~1, pp. 255--312, Jan. 2007.

\bibitem{Bertsekas1999}
D. P. Bertsekas. {\sl Nonlinear Programming}. Athena Scientific,
second edition, 1999.

\bibitem{Bertsekas1989}
D. P. Bertsekas, and J. N. Tsitsiklis. {\sl Parallel and Distributed
Computation: Numerical Methods}. Englewood Cliffs, NJ:
Prentice-Hall, 1989.

\bibitem{Lin2006}
X. Lin, and N. Shroff, ``Utility maximization for communication
networks with multipath routing,'' {\sl IEEE Trans. Auto. Control},
vol. 51, no. 5, pp. 766-781, 2006.

\bibitem{Etkin2007}
R. Etkin, A. Parekh, and D. Tse, ``Spectrum sharing for unlicensed bands,'' {\sl IEEE J. Sel. Areas
Commun.}, vol. 25, no. 3, pp. 517-528, Apr. 2007.

\bibitem{Schaerf1999}
A. Schaerf, ``A survey of automated timetabling,'' {\sl Artificial Intelligence Review}, vol. 13, pp. 87-127, 1999.

\bibitem{Qu2009}
R. Qu, E. K. Burke, B. McCollum, L.T.G. Merlot, and S.Y. Lee, ``A survey of search methodologies and automated system development for examination timetabling,''
{\sl Journal of Scheduling}, vol. 12, no. 1, pp. 55-89, 2009.

\bibitem{Luo2008}
Z.-Q. Luo and S. Zhang, ``Dynamic Spectrum Management: Complexity and Duality,'' {\sl IEEE J. Sel. Topics Signal Process}, vol. 2, no. 1, pp. 57-73, Feb. 2008.

\bibitem{Boyd2004}
S. Boyd, and L. Vandenberghe, {\sl Convex Optimization}. Cambridge,
UK: Cambridge University Press, 2004.

\bibitem{Lin2004}
X. Lin, and N. B. Shroff, ``Utility Maximization for Communication
Networks with Multi-path Routing,'' Technical Report, Purdue
University. Availeble:
{http://min.ecn.purdue.edu/{\_}linx/papers.html}, 2004.

\bibitem{3GPP}
3GPP TR 25.996 v9.0.0. ``Spatial channel model for
Multiple Input Multiple Output (MIMO) simulations''. 2009.
\end{thebibliography}
\end{document}